\documentclass[12pt]{article}

\usepackage{epsfig,amsmath,amssymb,amsfonts,amstext,amsthm}
\usepackage{latexsym,graphics,epsf,epsfig,color}
\usepackage{cite,url}

\newtheorem{theorem}{Theorem}
\newtheorem{prop}{Proposition}
\newtheorem{lemma}{Lemma}
\newtheorem{coro}[theorem]{Corollary}

\begin{document}

\title{Secure Symmetrical Multilevel Diversity Coding}

\author{Anantharaman~Balasubramanian, Hung D. Ly, Shuo Li,\\ Tie~Liu, and Scott~L.~Miller\thanks{This paper was presented in part at the 2009 Information Theory and Applications (ITA) Workshop, San Diego, CA, February~2009. This research was supported in part by the National Science Foundation under Grants CCF-08-45848 and CCF-09-16867. A.~Balasubramanian is with Qualcomm Inc., San Diego, CA 92121, USA (email: anantha.tamu@gmail.com). H. D. Ly, S. Li, T. Liu and S. L. Miller are with the Department of Electrical and Computer Engineering, Texas A\&M University, College Station, TX 77843, USA (email: \{hungly,lishuoxa\}@tamu.edu, \{tieliu,smiller\}@ece.tamu.edu).}
}

\date{\today}

\maketitle

\begin{abstract}
Symmetrical Multilevel Diversity Coding (SMDC) is a network compression problem introduced by Roche (1992) and Yeung (1995). In this setting, a simple separate coding strategy known as superposition coding was shown to be optimal in terms of achieving the minimum sum rate (Roche, Yeung, and Hau, 1997) and the entire admissible rate region (Yeung and Zhang, 1999) of the problem. This paper considers a natural generalization of SMDC to the secure communication setting with an additional eavesdropper. It is required that all sources need to be kept perfectly secret from the eavesdropper as long as the number of encoder outputs available at the eavesdropper is no more than a given threshold. First, the problem of encoding individual sources is studied. A precise characterization of the entire admissible rate region is established via a connection to the problem of secure coding over a three-layer wiretap network and utilizing some basic polyhedral structure of the admissible rate region. Building on this result, it is then shown that superposition coding remains optimal in terms of achieving the minimum sum rate for the general secure SMDC problem.
\end{abstract}

\section{Introduction}
Symmetrical Multilevel Diversity Coding (SMDC) is a source coding problem with $L$ independent discrete memoryless sources $(S_1,\ldots,S_L)$, where the \emph{importance} of the sources is assumed to decrease with the subscript $l$. The sources are to be encoded by a total of $L$ encoders, where the rate of the $l$th encoder output is $R_l$. The decoder can access a subset $U \subseteq \Omega_L:=\{1,\ldots,L\}$ of the encoder outputs. Which subset of the encoder outputs is available at the decoder is \emph{unknown} a priori at the encoders. However, no matter which subset $U$ actually realizes, the sources $(S_1,\ldots,S_m)$ need to be asymptotically perfectly reconstructed at the decoder whenever $|U|\geq m$. Note that the word ``symmetrical" here refers to the fact that the sources that need to be reconstructed at the decoder depend on the available subset of the encoder outputs only via its cardinality. The rate allocations at different encoders, however, can be different and are not necessarily symmetrical.

The problem of Multilevel Diversity Coding (MDC) was introduced by Roche \cite{Roc-Thesis92} and Yeung \cite{Yeu-IT95} in the early 1990s. In particular, \cite{Yeu-IT95} considered the simple coding strategy of separately encoding different sources at the encoders, subsequently referred to as \emph{superposition coding}. The aforementioned SMDC problem was first systematically studied in \cite{RYH-IT97}, where it was shown that superposition coding can achieve the minimum sum rate for the general SMDC problem (with an arbitrary total number of encoders $L$) and the entire admissible rate region with $L=3$ encoders. The problem regarding whether superposition coding can achieve the entire admissible rate region for the general SMDC problem, however, remained open. Finally, in a very elegant (albeit highly technical) paper \cite{YZ-IT99}, Yeung and Zhang resolved the open problem by positive through the so-called \emph{$\alpha$-resolution} method.

Recent years have seen a flurry of research on information-theoretic security. See \cite{LPS-M09} and \cite{LT-M10} for surveys of recent progress in this field. Motivated by this renewed interest, in this paper we consider the problem of \emph{Secure} Symmetrical Multilevel Diversity Coding (S-SMDC) in the presence of an additional eavesdropper. Specifically, a collection of $L-N$ independent discrete memoryless sources $(S_1,\ldots,S_{L-N})$ are to be encoded by a total of $L$ encoders, where the rate of the $l$th encoder output is $R_l$. A legitimate receiver and an eavesdropper can access a subset $U \subseteq \Omega_L$ and $A \subseteq \Omega_L$ of the encoder outputs, respectively. Which subsets of the encoder outputs are available at the legitimate receiver and the eavesdropper are \emph{unknown} a priori at the encoders. However, no matter which subsets $U$ and $A$ actually occur, the sources $(S_1,\ldots,S_k)$ need to be asymptotically perfectly reconstructed at the legitimate receiver whenever $|U|\geq N+k$, and the entire collection of the sources $(S_1,\ldots,S_{L-N})$ needs to be kept \emph{perfectly} secret from the eavesdropper as long as $|A| \leq N$. As before, the word ``symmetrical" here refers to the access structure at the legitimate receiver and the eavesdropper, but not to the rate allocations at different encoders. We envision that such a communication scenario is useful for designing distributed information storage systems \cite{Roc-Thesis92} where information retrieval needs to be both robust and secure. 

As mentioned previously, separate encoding of different sources (superposition coding) can achieve the entire admissible rate region for the general SMDC problem without any secrecy constraints \cite{YZ-IT99}. It is thus natural to ask whether the same separate encoding strategy would remain optimal for the general S-SMDC problems. For the classical SMDC problems without any secrecy constraints, the problem of efficient encoding of individual sources is essentially to transmit the source over an \emph{erasure} channel and is well understood based on the earlier work of Singleton \cite{Sin-IT64}. For the S-SMDC problems, however, the problem of efficient encoding of individual sources is closely related to the problem of secure coding over a \emph{Wiretap Network (WN)} \cite{CY-ITS}, which, in its most general setting, is a very challenging problem in information-theoretic security.

The rest of the paper is organized as follows. In Sec.~\ref{sec:SL}, we focus on the problem of encoding individual sources, i.e., Secure Symmetrical Single-level Diversity Coding (S-SSDC). By leveraging the results of \cite{CY-ITS} on secure coding over a three-layer WN and utilizing some basic polyhedral structure of the admissible rate region, we provide a precise characterization of the entire admissible rate region for the general S-SSDC problem. Building on this result, in Sec.~\ref{sec:ML} we show that superposition coding can achieve the minimum sum rate for the general S-SMDC problem. Finally, in Sec.~\ref{sec:Con} we conclude the paper with some remarks.

\section{Secure Symmetrical Single-Level Diversity Coding}\label{sec:SL}
\subsection{Problem Statement}\label{sec:SL-ps}
Let $\{S[t]\}_{t=1}^{\infty}$ be a discrete memoryless source with time index $t$ and let $S^n:=(S[1],\ldots,S[n])$. An $(L,N,m)$ S-SSDC problem consists of a set of $L$ encoders, a legitimate receiver who has access to a subset $U \subseteq \Omega_L$ of the encoder outputs, and an eavesdropper who has access to a subset $A \subseteq \Omega_L$ of the encoder outputs. Which subsets of the encoder outputs are available at the legitimate receiver and the eavesdropper are \emph{unknown} a priori at the encoders. However, no matter which subsets $U$ and $A$ actually occur, the legitimate receiver must be able to asymptotically perfectly reconstruct the source whenever $|U| \geq m$, and the source must be kept \emph{perfectly} secret from the eavesdropper as long as $|A| \leq N$. Obviously, reliable and secure communication of the source is possible only when $m > N$.

Formally, an $(n,(M_1,\ldots,M_L))$ code is defined by a collection of $L$ encoding functions
\begin{equation}
e_l: \mathcal{S}^n \times \mathcal{K} \rightarrow \{1,\ldots,M_l\}, \quad \forall l=1,\ldots,L
\end{equation}
and decoding functions
\begin{equation}
d_U: \prod_{l \in U}\{1,\ldots,M_l\} \rightarrow \mathcal{S}^n, \quad \forall U \subseteq \Omega_L \; \mbox{s.t.} \; |U| \geq m.
\end{equation}
Here, $\mathcal{K}$ denotes the key space accessible to all $L$ encoders. There are no limitations on the size of the key space $\mathcal{K}$. However, the secret key is only shared by the encoders, but \emph{not} with the legitimate receiver or the eavesdropper. A nonnegative rate tuple $(R_1,\ldots,R_L)$ is said to be \emph{admissible} if for every $\epsilon>0$, there exits, for sufficiently large block length $n$, an $(n,(M_1,\ldots,M_L))$ code such that:
\begin{itemize}
\item (Rate constraints)
\begin{equation}
\frac{1}{n}\log M_l \leq R_l +\epsilon, \quad \forall l =1,\ldots,L;\label{eq:S-Rate}
\end{equation}
\item (Asymptotically perfect reconstruction at the legitimate receiver)
\begin{equation}
\mathrm{Pr}\{d_U(X_U) \neq S^n\} \leq \epsilon, \quad \forall U \subseteq \Omega_L \; \mbox{s.t.} \; |U| \geq m
\label{eq:S-Rec}
\end{equation}
where $X_l:=e_l(S^n,K)$ is the output of the $l$th encoder, $K$ is the secret key shared by all $L$ encoders, and $X_U:=\{X_l:  l \in U\}$; and
\item (Perfect secrecy at the eavesdropper)
\begin{equation}
H(S^n|X_A) = H(S^n), \qquad \forall A\subseteq \Omega_L \; \mbox{s.t.} \; |A|\leq N
\label{eq:S-Per}
\end{equation}
i.e., observing the encoder outputs $X_A$ does not provide \emph{any} information regarding to the source sequence $S^n$. 
\end{itemize}

The \emph{admissible rate region} $\mathcal{R}$ is the collection of \emph{all} admissible rate tuples $(R_1,\ldots,R_L)$. The \emph{minimum sum rate} $R_{ms}$ is defined as
\begin{equation}
R_{ms} := \min_{(R_1,\ldots,R_L) \in \mathcal{R}}\sum_{l=1}^{L}R_l.
\end{equation}

\subsection{Main Results}\label{sec:SL-main}
The following lemma provides a simple outer bound on the admissible rate region of the general S-SSDC problem. Let $\mathcal{R}(L,k,H)$ be the collection of all nonnegative rate tuples $(R_1,\ldots,R_L)$ satisfying
\begin{equation}
\sum_{l \in D} R_l \geq H, \quad \forall D \in \Omega_L^{(k)}
\label{eq:S-SSDC}
\end{equation}
where $\Omega_L^{(k)}$ is the collection of all subsets of $\Omega_L$ of size $k$.

\begin{lemma}\label{lemma:S-SSDC}
For any $(L,N,m)$ S-SSDC problem, the admissible rate region
\begin{equation}
\mathcal{R} \subseteq \mathcal{R}(L,m-N,H(S)).
\end{equation}
\end{lemma}

Lemma~\ref{lemma:S-SSDC} can be proved using standard information-theoretic techniques. For completeness, a proof is included in Appendix~\ref{app:A}. The above outer bound is known to be tight in the following two special cases:
\begin{itemize}
\item[1)] When $N=0$, the $(L,N,m)$ S-SSDC problem reduces to the classical $(L,m)$ SSDC problem without any secrecy constraints, for which the admissible rate region is known \cite{Sin-IT64} to be $\mathcal{R}(L,m,H(S))$.
\item[2)] With $N>0$ but $m=N+1$, a collection $D$ of the encoder outputs will either lead to an asymptotically perfect reconstruction of the source (whenever $|D|\geq N+1$), or provide zero information on the source (whenever $|D| \leq N$). In this case, the $(L,N,m)$ S-SSDC problem reduces to the classical $(L,N)$ \emph{threshold secret sharing} problem, for which the admissible rate region is known \cite{Sha-CACM79,Bla-NCC79} to be $\mathcal{R}(L,1,H(S))$.
\end{itemize}

The main result of this section is that the outer bound $\mathcal{R}(L,m-N,H(S))$ is in fact the admissible rate region for the \emph{general} S-SSDC problem, as summarized in the following theorem.

\begin{theorem}\label{thm:S-SSDC}
For any $(L,N,m)$ S-SSDC problem, the admissible rate region
\begin{equation}
\mathcal{R}=\mathcal{R}(L,m-N,H(S)).
\end{equation}
\end{theorem}

A proof of the theorem is provided in Sec.~\ref{sec:SL-pf}. To show that \emph{every} rate tuple in $\mathcal{R}(L,m-N,H(S))$ is admissible, our proof proceeds in the following two steps. First, we show that for any $(L,N,m)$ S-SSDC problem, the symmetrical rate tuple $(H(S)/(m-N),\ldots,H(S)/(m-N))$ is admissible. In our proof, this is accomplished by relating the S-SSDC problem to the problem of secure coding over a \emph{three-layer} WN and using the result of \cite[Th.~3]{CY-ITS} on an achievable secrecy rate for the generic WN. Building on the previous result, next we show that every rate tuple in $\mathcal{R}(L,m-N,H(S))$ is admissible via an induction argument (inducting on the total number of encoders $L$) and the following polyhedral structure of $\mathcal{R}(L,k,H)$.
\begin{prop}\label{lemma:poly}
$\mathcal{R}(L,k,H)$ is a pointed polyhedron in $\mathbb{R}^L$ with the following structural properties:
\begin{enumerate}
\item[1)] The characteristic cone of $\mathcal{R}(L,k,H)$ is given by 
$\{(R_1,\ldots,R_L): R_l \geq 0, \; \forall l=1,\ldots,L\}$.
\item[2)] Among all corner points (vertices) of $\mathcal{R}(L,k,H)$,
$(H/k,\ldots,H/k)$ is the \emph{only} one with all \emph{strictly} positive entries (if there exists any).
\item[3)] For any $l=1,\ldots,L$, the $R_l=0$ slice of $\mathcal{R}(L,k,H)$ is \emph{isomorphic} to $\mathcal{R}(L-1,k-1,H)$. In particular, the $R_L=0$ slice of $\mathcal{R}(L,k,H)$ is \emph{identical} to $\mathcal{R}(L-1,k-1,H)$, i.e.,
\begin{equation}
\{(R_1,\ldots,R_{L-1}): (R_1,\ldots,R_{L-1},0) \in \mathcal{R}(L,k,H)\} = \mathcal{R}(L-1,k-1,H).
\end{equation}
\end{enumerate}
\end{prop}

\begin{proof}
Property 1 follows directly from the definition of characteristic cone \cite[Lec.~2]{Che-L10}. Property 2 is due to the fact that 
\begin{equation}
(R_1,\ldots,R_L) = \left(H/k,\ldots,H/k\right)
\end{equation}
is a solution to the equations
\begin{equation}
\sum_{l \in D} R_l = H, \quad \forall D \in \Omega_L^{(k)}.
\end{equation}

To see property 3, note that the $R_l=0$ slice of $\mathcal{R}(L,k,H)$ is given by all nonnegative rate tuples $(R_1,\ldots,R_{l-1},R_{l+1},\ldots,R_{L})$ satisfying
\begin{equation}
\sum_{d \in D}{R_d} \geq H, \quad \forall D \in \Omega_{L\setminus \{l\}}^{(k-1)} \cup \Omega_{L\setminus \{l\}}^{(k)}
\end{equation}
where $\Omega_{L\setminus \{l\}}^{(k)}$ denotes all subsets of $\Omega_L \setminus \{l\}$ of size $k$. Since every inequality with $D \in \Omega_{L\setminus \{l\}}^{(k)}$ is dominated by every inequality with $D' \in \Omega_{L\setminus \{l\}}^{(k-1)}$ and such that $D' \subseteq D$, we have the desired property. 
\end{proof}

\begin{figure}[t]
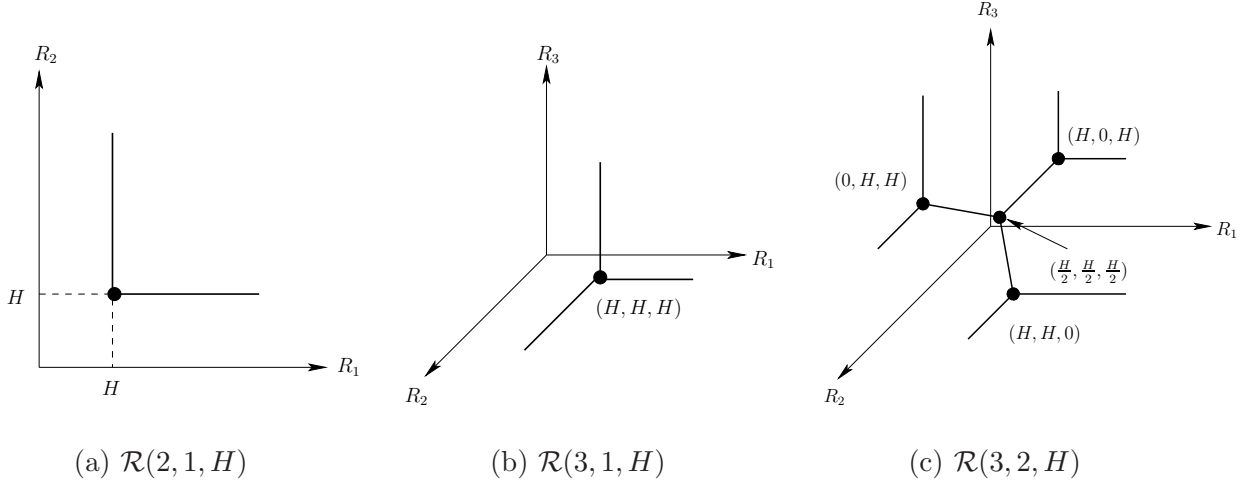

  \begin{minipage}[b]{0.32\linewidth}
  \centerline{\scalebox{0.65}{\input{l2n1m2.pstex_t}}}
  \vspace{0.3cm}
  \centerline{\mbox{\small (a) $\mathcal{R}(2,1,H)$}}
  \vspace{0.8cm}
  \end{minipage}
  \begin{minipage}[b]{0.32\linewidth}
  \centerline{\scalebox{0.65}{\input{l3n1m2.pstex_t}}}
  \vspace{0.3cm}
  \centerline{\mbox{\small (b) $\mathcal{R}(3,1,H)$}}
  \vspace{0.8cm}
  \end{minipage}
  \begin{minipage}[b]{0.32\linewidth}
  \centerline{\scalebox{0.6}{\input{l3n1m3.pstex_t}}}
  \vspace{0.3cm}
  \centerline{\mbox{\small (c) $\mathcal{R}(3,2,H)$}}
  \vspace{0.8cm}
  \end{minipage}
  \caption{Illustration of the rate region $\mathcal{R}(L,k,H)$ for $L=2$ and 3. The $R_l=0$ slices of $\mathcal{R}(3,1,H)$ are empty sets, and the $R_l=0$ slices of $\mathcal{R}(3,2,H)$ are isomorphic to $R(2,1,H)$.}
  \label{fig:poly}
\end{figure}

Fig.~\ref{fig:poly} illustrates the above polyhedral structure of $\mathcal{R}(L,k,H)$ for $L=2$ and $3$. The following corollary summarizes the minimum sum rate for the general S-SSDC problem.
\begin{coro}\label{cor:S-SSDC}
For any $(L,N,m)$ S-SSDC problem, the minimum sum rate
\begin{equation}
R_{ms} = \frac{L}{m-N}H(S).
\end{equation}
\end{coro}

\begin{proof}
Let us first verify that 
\begin{equation}
\min_{(R_1,\ldots,R_L) \in \mathcal{R}(L,k,H)}\sum_{l=1}^{L}R_l = \frac{L}{k}H.
\label{S0}
\end{equation} 
For any rate tuple $(R_1,\ldots,R_L) \in \mathcal{R}(L,k,H)$, we have
\begin{equation}
\sum_{l \in D}R_l \geq H, \quad \forall D \in \Omega_L^{(k)}.
\end{equation}
Summing over all $D \in \Omega_L^{(k)}$ gives
\begin{equation}
\sum_{D \in \Omega_L^{(k)}}\sum_{l \in D}R_l = 
\left(
\begin{array}{c}
  L-1\\
  k-1   
\end{array}
\right)\sum_{l=1}^{L}R_l
\geq 
\left(
\begin{array}{c}
  L\\
  k   
\end{array}
\right)H.
\end{equation}
We thus have
\begin{equation}
\sum_{l=1}^{L}R_l \geq 
\frac{
\left(
\begin{array}{c}
  L\\
  k   
\end{array}
\right)}
{\left(
\begin{array}{c}
  L-1\\
  k-1   
\end{array}
\right)}H=\frac{L}{k}H
\label{S1}
\end{equation}
for any rate tuple $(R_1,\ldots,R_L) \in \mathcal{R}(L,k,H)$. On the other hand, note that the symmetrical rate tuple
\begin{equation}
\left(H/k,\ldots,H/k\right) \in \mathcal{R}(L,k,H)
\end{equation}
so
\begin{equation}
\min_{(R_1,\ldots,R_L) \in \mathcal{R}(L,k,H)}\sum_{l=1}^{L}R_l \leq \frac{L}{k}H.
\label{S2}
\end{equation} 
Combining \eqref{S1} and \eqref{S2} completes the proof of \eqref{S0}. 

Now by Theorem~\ref{thm:S-SSDC},  
\begin{equation}
R_{ms} = \min_{(R_1,\ldots,R_L) \in \mathcal{R}}\sum_{l=1}^{L}R_l=\min_{(R_1,\ldots,R_L) \in \mathcal{R}(L,m-N,H(S))}\sum_{l=1}^{L}R_l = \frac{L}{m-N}H(S).
\end{equation} 
This completes the proof of the corollary.
\end{proof}

\subsection{Proof of Theorem~\ref{thm:S-SSDC}}\label{sec:SL-pf}
Let us first show that the symmetrical rate tuple $(H(S)/(m-N),\ldots,H(S)/(m-N))$ is admissible by considering the following simple \emph{source-channel separation} scheme for the $(L,N,m)$ S-SSDC problem:
\begin{itemize}
\item First compress the source sequence $S^n$ into a source message $W$ using a fixed-length lossless source code. It is well known \cite[Ch.~3.2]{CT-B06} that the rate $R$ of the source message $W$ can be made arbitrarily close to the entropy rate $H(S)$ for sufficiently large block length $n$. 
\item Next, the source message $W$ is delivered to the legitimate receiver using a secure $$(L,N,m,(R_1,\ldots,R_L))$$ WN code. 
\end{itemize} 

The problem of secure coding over a WN was formally introduced in \cite{CY-ITS}. A generic WN $(\mathcal{G},s,\mathcal{U},\mathcal{A})$ consists of a directed acyclic network $\mathcal{G}$, a source node $s$, a set of user nodes $\mathcal{U}$, and a collection of sets of wiretapped edges $\mathcal{A}$. Each member of $\mathcal{A}$ may be fully accessed by an eavesdropper, but no eavesdropper may access more than one member of $\mathcal{A}$. The source node has access to a message $W$, which is intended for all user nodes in $\mathcal{U}$ but needs to be kept \emph{perfectly} secret from the eavesdroppers. The maximum achievable secrecy rate for $W$ is called the \emph{secrecy capacity} of the WN and is denoted by $C_s(\mathcal{G},s,\mathcal{U},\mathcal{A})$.

\begin{figure}[t]
\centerline{\scalebox{0.8}{\input{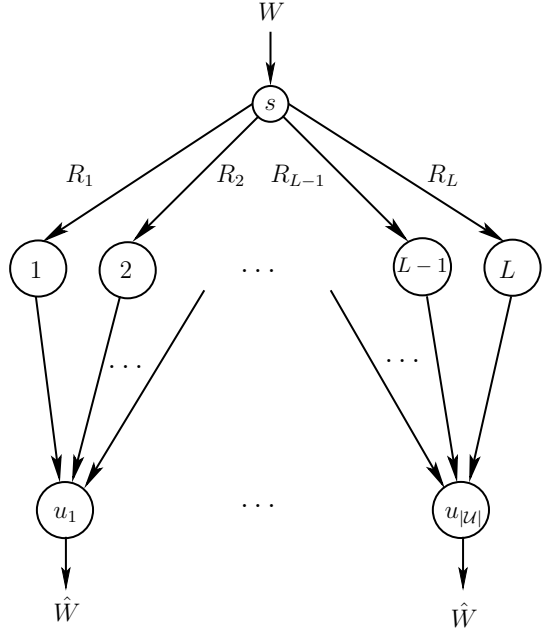}}}
\caption{Illustration of the $(L,N,m,(R_1,\ldots,R_L))$ WN.}
\label{fig:WN}
\end{figure}

An $(L,N,m,(R_1,\ldots,R_L))$ WN is a special WN with three layers of nodes: top, middle, and bottom. As illustrated in Fig.~\ref{fig:WN}, the only node in the top layer is the source node $s$. There are $L$ intermediate nodes in the middle layer, each corresponding to an encoder in the $(L,N,m)$ S-SSDC problem. For each $l=1,\ldots,L$, the source node $s$ is connected to the intermediate node $l$ by a channel $(s, l)$ with capacity $R_l$. There are
\begin{equation}
|\mathcal{U}|=\left(
\begin{array}{c}
  L   \\
  m   
\end{array}
\right)
\end{equation} 
user nodes in the bottom layer, each corresponding to a possible realization of the legitimate receiver in the $(L,N,m)$ S-SSDC problem and is connected to $m$ intermediate nodes through $m$ infinite-capacity channels. Finally, the collection of sets of wiretapped edges $\mathcal{A}$ is defined as 
\begin{equation}
\mathcal{A} := \left\{\{(s, l)|l \in A\}: A \in \Omega_L^{(N)}\right\}
\end{equation}
where each set of wiretapped edges in $\mathcal{A}$ corresponds to a possible realization of the eavesdropper in the $(L,N,m)$ S-SSDC problem.

Based on the aforementioned connection between the $(L,N,m)$ S-SSDC problem and the problem of secure coding over the $(L,N,m,(R_1,\ldots,R_L))$ WN, we have the following simple lemma.

\begin{lemma}\label{lemma:sep}
A nonnegative rate tuple $(R_1,\ldots,R_L)$ is admissible for the $(L,N,m)$ S-SSDC problem if the entropy rate of the source is less than or equal to the secrecy capacity of the $(L,N,m,(R_1,\ldots,R_L))$ WN, i.e.,
\begin{equation}
H(S) \leq C_s(L,N,m,(R_1,\ldots,R_L)).
\end{equation}
\end{lemma}

In general, characterizing the exact secrecy capacity of a WN can be very difficult. For a generic WN $(\mathcal{G},s,\mathcal{U},\mathcal{A})$, the following secrecy rate 
\begin{equation}
R_s = \min_{u \in \mathcal{U}, A \in \mathcal{A}}\left[mincut(s,u)-mincut(s,A)\right]
\end{equation}
is known \cite{CY-ITS} to be achievable. Here, $mincut(s,u)$ denotes the value of a minimum cut between the source node $s$ and the user node $u$, and $mincut(s,A)$ denotes the value of a minimum cut between the source node $s$ and the set of wiretapped edges $A$. For the $(L,N,m,(H(S)/(m-N),\ldots,H(S)/(m-N)))$ WN, it is straightforward to verify that 
\begin{equation}
mincut(s,u) = \frac{m}{m-N}H(S), \quad \forall u \in \mathcal{U}
\end{equation}
and
\begin{equation}
mincut(s,A) = \frac{N}{m-N}H(S), \quad \forall A \in \mathcal{A}.
\end{equation}
Hence, the secrecy rate
\begin{equation}
R_s = \frac{m}{m-N}H(S)-\frac{N}{m-N}H(S)=H(S)
\end{equation}
is achievable for the $(L,N,m,(H(S)/(m-N),\ldots,H(S)/(m-N)))$ WN. We summarize this result in the following lemma. 
\begin{lemma}\label{lemma:cy}
For any $(L,N,m,(H(S)/(m-N),\ldots,H(S)/(m-N)))$ WN, the secrecy capacity can be bounded from below as
\begin{equation}
C_s(L,N,m,(H(S)/(m-N),\ldots,H(S)/(m-N))) \geq H(S).
\end{equation}
\end{lemma}

Combining Lemmas~\ref{lemma:sep} and \ref{lemma:cy} proves the admissibility of the symmetrical rate tuple 
$$\left(H(S)/(m-N),\ldots,H(S)/(m-N)\right).$$

Building on the previous result, next we show that \emph{every} rate tuple in $\mathcal{R}(L,m-N,H(S))$ is admissible. By Proposition~\ref{lemma:poly},  $\mathcal{R}(L,m-N,H(S))$ is a pointed polyhedron with the characteristic cone given by $\{(R_1,\ldots,R_{L}): R_l \geq 0, \; \forall l=1,\ldots,L\}$. Thus, to show that all rate tuples in $\mathcal{R}(L,m-N,H(S))$ are admissible, it is sufficient to show that all \emph{corner points} of $\mathcal{R}(L,m-N,H(S))$ are admissible.

We shall consider proof by induction, where the induction is on the total number of encoders $L$. First consider the base case with $L=2$. When $L=2$, there is only one nontrivial $(L,N,m)$ S-SSDC problem: the $(2,1,2)$ S-SSDC problem. Note that the rate region $\mathcal{R}(2,1,H(S))$ has only one corner point: the symmetrical rate pair $(H(S),H(S))$, whose admissibility has already been established. We thus conclude that every rate tuple in $\mathcal{R}(2,1,H(S))$ is admissible for the $(2,1,2)$ S-SSDC problem. 

Now assume that for every nontrivial $(L-1,N',m')$ S-SSDC problem, all rate tuples in $\mathcal{R}(L',m'-N',H(S))$ are admissible. Based on this assumption, next we show that all corner points of $\mathcal{R}(L,m-N,H(S))$ are admissible for the $(L,N,m)$ S-SSDC problem.  We shall consider the corner points with all \emph{strictly} positive entries (it they exist) and those with \emph{at least one zero} entry separately:
\begin{itemize}
\item[1)] By Proposition~\ref{lemma:poly}, the symmetrical rate tuple $(H(S)/(m-N),\ldots,H(S)/(m-N))$ is the \emph{only} corner point of $\mathcal{R}(L,m-N,H(S))$ with all \emph{strictly} positive entries (if it exists), whose admissibility has already been established.
\item[2)] To prove the admissibility of the corner points of $\mathcal{R}(L,m-N,H(S))$ with \emph{at least one zero} entry, by the \emph{symmetry} of the rate region $\mathcal{R}(L,m-N,H(S))$ we may consider without loss of generality those with $R_L=0$. Note that if an $(n,(M_1,\ldots,M_{L-1}))$ code satisfies both asymptotically perfect reconstruction constraint \eqref{eq:S-Rec} and the perfect secrecy constraint \eqref{eq:S-Per} for the $(L-1,N,m-1)$ S-SSDC problem, an $(n,(M_1,\ldots,M_{L-1},1))$ code with the \emph{same} encoding functions for encoders 1 to $L-1$ (encoder $L$ uses a constant encoding function) also satisfies (trivially) both constraints for the $(L,N,m)$ S-SSDC problem. Thus, if $(R_1,\ldots,R_{L-1})$ is an admissible rate tuple for the $(L-1,N,m-1)$ S-SSDC problem, then $(R_1,\ldots,R_{L-1},0)$ is also an admissible rate tuple for the $(L,N,m)$ S-SSDC problem. By the induction assumption, all rate tuples in $\mathcal{R}(L-1,m-N-1,H(S))$ are admissible for the $(L-1,N,m-1)$ problem. Combined with Proposition~\ref{lemma:poly}, this implies that all rate tuples in 
\begin{equation}
\begin{array}{l}
\{(R_1,\ldots,R_{L-1},0): (R_1,\ldots,R_{L-1}) \in \mathcal{R}(L-1,m-N-1,H(S))\}\\
\hspace{120pt}  = \; \{(R_1,\ldots,R_L) \in \mathcal{R}(L,m-N,H(S)): R_L=0\}
\end{array}
\end{equation}
i.e., the $R_L=0$ slice of $\mathcal{R}(L,m-N,H(S))$, are admissible for the $(L,N,m)$ S-SSDC problem. As a special case, all corner points of $\mathcal{R}(L,m-N,H(S))$ with $R_L=0$ are admissible for the $(L,N,m)$ S-SSDC problem. 
\end{itemize}

Combining Steps 1 and 2 proves that all corner points of $\mathcal{R}(L,m-N,H(S))$ are admissible. We thus conclude that all rate tuples in $\mathcal{R}(L,m-N,H(S))$ are admissible. This completes the induction step and hence the proof of the theorem.

\section{Secure Symmetrical Multilevel Diversity Coding}\label{sec:ML}
\subsection{Problem Statement}\label{sec:ML-ps}
Let $\{S_1[t],\ldots,S_{L-N}[t]\}_{t=1}^{\infty}$ be a collection of $L-N$ independent discrete memoryless sources with time index $t$ and let $S_k^n:=(S_k[1],\ldots,S_k[n])$. An $(L,N)$ S-SMDC problem consists of a set of $L$ encoders, a legitimate receiver who has access to a subset $U$ of the encoder outputs, and an eavesdropper who has access to a subset $A$ of the encoder outputs. Which subsets of the encoder outputs are available at the legitimate receiver and the eavesdropper are \emph{unknown} a priori at the encoders. However, no matter which subsets $U$ and $A$ actually occur, the legitimate receiver must be able to asymptotically perfectly reconstruct the sources $(S_1,\ldots,S_k)$ whenever $|U|=N+k$, and all sources $(S_1,\ldots,S_{L-N})$ must be kept perfectly secure from the eavesdropper as long as $|A|\leq N$. 

Formally, an $(n,(M_1,\ldots,M_L))$ code is defined by a collection of $L$ encoding functions
\begin{equation}
e_l: \prod_{k=1}^{L-N}\mathcal{S}_k^n \times \mathcal{K} \rightarrow \{1,\ldots,M_l\}, \quad \forall l=1,\ldots,L
\end{equation}
and decoding functions
\begin{equation}
d_U: \prod_{l \in U}\{1,\ldots,M_l\} \rightarrow \prod_{k=1}^{|U|-N}\mathcal{S}_k^n, \quad \forall U \subseteq \Omega_L \; \mbox{s.t.} \; |U| \geq N+1.
\end{equation}
Here, $\mathcal{K}$ is the key space accessible to all $L$ encoders. A nonnegative rate tuple $(R_1,\ldots,R_L)$ is said to be \emph{admissible} if for every $\epsilon>0$, there exits, for sufficiently large block length $n$, an $(n,(M_1,\ldots,M_L))$ code such that:
\begin{itemize}
\item (Rate constraints)
\begin{equation}
\frac{1}{n}\log M_l \leq R_l +\epsilon, \quad \forall l =1,\ldots,L;
\label{eq:M-Rate}
\end{equation}
\item (Asymptotically perfect reconstruction at the legitimate receiver)
\begin{equation}
\mathrm{Pr}\{d_U(X_U) \neq (S_1^n,\ldots,S_{|U|-N}^n)\} \leq \epsilon, \quad \forall U \subseteq \Omega_L \; \mbox{s.t.} \; |U| \geq N+1
\label{eq:M-Rec}
\end{equation}
where $X_l:=e_l((S_1^n,\ldots,S_{L-N}^n),K)$ is the output of the $l$th encoder, $K$ is the secret key shared by all $L$  encoders, and $X_U:=\{X_l: l \in U\}$; and
\item (Perfect secrecy at the eavesdropper)
\begin{equation}
H(S_1^n,\ldots,S_{L-N}^n|X_A) = H(S_1^n,\ldots,S_{L-N}^n), \quad \forall A \subseteq \Omega_L \; \mbox{s.t.} \; |A|\leq N
\label{eq:M-Per}
\end{equation}
i.e., observing the encoder outputs $X_A$ does not provide \emph{any} information regarding to the source sequences $(S_1^n,\ldots,S_{L-N}^n)$. 
\end{itemize}

\subsection{Main Results}\label{sec:ML-main}
Motivated by the success of \cite{Yeu-IT95,RYH-IT97,YZ-IT99} on the classical SMDC problem without any secrecy constraints, here we focus on superposition coding where the output of the $l$th encoder $X_l$ is given by
\begin{equation}
X_l=\left(X_l^{(1)},\ldots,X_l^{(L-N)}\right)
\label{eq:SP}
\end{equation}
and $X_l^{(k)}$ is the coded message for source $S_k$ at the $l$th encoder using an $(L,N,N+k)$ S-SSDC code. Note here that all sources are encoded \emph{separately} at the encoders, and there is \emph{no} coding across different sources. Thus, if $(R_1^{(k)},\ldots,R_L^{(k)})$ is an admissible rate tuple for the $(L,N,N+k)$ S-SSDC problem with source $S_k$, then the rate tuple
\begin{equation}
(R_1,\ldots,R_L)=\left(\sum_{k=1}^{L-N}R_1^{(k)},\ldots,\sum_{k=1}^{L-N}R_L^{(k)}\right)
\label{eq:SC}
\end{equation}
is admissible for the $(L,N)$ S-SMDC problem. 

By Corollary~\ref{cor:S-SSDC}, the minimum sum rate for the $(L,N,N+k)$ S-SSDC problem with source $S_k$ is given by $(L/k)H(S_k)$. It follows that 
$\sum_{k=1}^{L-N}(L/k)H(S_k)$ is the minimum sum rate that can be achieved by superposition coding for the $(L,N)$ S-SMDC problem. The main result of this section is that $\sum_{k=1}^{L-N}(L/k)H(S_k)$ is in fact the minimum sum rate that can be achieved by \emph{any} coding scheme for the $(L,N)$ S-SMDC problem. Thus, superposition coding is optimal in terms of achieving the minimum sum rate for the general S-SMDC problem. We summarize this result in the following theorem.

\begin{theorem}\label{thm:S-SMDC1}
Superposition coding can achieve the minimum sum rate for the general $(L,N)$ S-SMDC problem, which is given by 
\begin{equation}
R_{ms} = \sum_{k=1}^{L-N}\frac{L}{k}H(S_{k}).
\end{equation}
\end{theorem}

A proof of the theorem is provided in Sec~\ref{sec:MS}. The proof uses an induction argument and is built on the classical subset inequality of Han \cite[Ch.~17.6]{CT-B06} and the following key proposition.

\begin{prop}\label{lemma:ML}
For any $(n,(M_1,\ldots,M_L))$ code that satisfies both asymptotically perfect reconstruction constraint \eqref{eq:M-Rec} and the perfect secrecy constraint \eqref{eq:M-Per}, we have
\begin{equation}
H(X_D|S_1^n,\ldots,S_{k-1}^n,X_A) \geq nH(S_k)+H(X_D|S_1^n,\ldots,S_k^n,X_A)-n\delta_k(n,\epsilon)
\end{equation}
where 
\begin{equation}
\delta_k(n,\epsilon) := 1/n+\epsilon\sum_{\alpha=1}^{k}\log|\mathcal{S}_\alpha|
\end{equation}
for any $A \in \Omega_L^{(N)}$ and $D \in \Omega_{L}^{(k)}$ such that $A \cap D =\emptyset$ and any $k=1,\ldots,L-N$.
\end{prop}

\subsection{Proof of the Main Results} \label{sec:MS}
Let us first prove Proposition~\ref{lemma:ML}. Since $|A|=N$, $|D|=k$, and $A\cap D =\emptyset$, we have $|D \cup A|=N+k$. For any $(n,(M_1,\ldots,M_L))$ code that satisfies both asymptotically perfect reconstruction constraint \eqref{eq:M-Rec} and the perfect secrecy constraint \eqref{eq:M-Per}, we have by Fano's inequality
\begin{equation}
H(S_1^n,\ldots,S_{k}^n|X_D,X_A)\leq n\delta_k(n,\epsilon) \label{T660}
\end{equation}
and
\begin{equation}
H(S_1^n,\ldots,S_{k}^n|X_A)=H(S_1^n,\ldots,S_{k}^n).\label{T661}
\end{equation}
Thus,
\begin{eqnarray}
&& H(X_D|S_1^n,\ldots,S_{k-1}^n,X_A)+n\delta_k(n,\epsilon)\nonumber\\
&& \hspace{20pt} \geq \; H(X_D|S_1^n,\ldots,S_{k-1}^n,X_A)+H(S_1^n,\ldots,S_{k}^n|X_D,X_A)
\label{T900}\\
&& \hspace{20pt} \geq \; H(X_D|S_1^n,\ldots,S_{k-1}^n,X_A)+H(S_{k}^n|S_1^n,\ldots,S_{k-1}^n,X_D,X_A)
\\
&& \hspace{20pt} = \; H(X_D,S_{k}^n|S_1^n,\ldots,S_{k-1}^n,X_A)\\
&& \hspace{20pt} = \; H(S_{k}^n|S_1^n,\ldots,S_{k-1}^n,X_A)+H(X_D|S_1^n,\ldots,S_{k}^n,X_A)\\
&& \hspace{20pt} = \; H(S_1^n,\ldots,S_{k}^n|X_A)-H(S_1^n,\ldots,S_{k-1}^n|X_A)+H(X_D|S_1^n,\ldots,S_{k}^n,X_A)\\
&& \hspace{20pt} = \; H(S_1^n,\ldots,S_{k}^n)-H(S_1^n,\ldots,S_{k-1}^n|X_A)+H(X_D|S_1^n,\ldots,S_{k}^n,X_A) \label{T901}\\
&& \hspace{20pt} \geq \; H(S_1^n,\ldots,S_{k}^n)-H(S_1^n,\ldots,S_{k-1}^n)+H(X_D|S_1^n,\ldots,S_{k}^n,X_A)\label{T902}\\
&& \hspace{20pt} = \; H(S_{k}^n|S_1^n,\ldots,S_{k-1}^n)+H(X_D|S_1^n,\ldots,S_{k}^n,X_A)\\
&& \hspace{20pt} = \; H(S_{k}^n)+H(X_D|S_1^n,\ldots,S_{k}^n,X_A) \label{T903}\\
&& \hspace{20pt} = \; nH(S_{k})+H(X_D|S_1^n,\ldots,S_{k}^n,X_A) \label{T904}
\end{eqnarray}
where \eqref{T900} follows from \eqref{T660}, \eqref{T901} follows from \eqref{T661}, \eqref{T902} follows from the fact that conditioning reduces entropy, \eqref{T903} follows from the fact that the sources $S_1,\ldots,S_{k}$ are mutually independent, and \eqref{T904} follows from the fact that the source $S_{k}$ is memoryless. Moving $n\delta_k(n,\epsilon)$ to the right-hand side of the inequality completes the proof of Proposition~\ref{lemma:ML}.

Building on the result of Proposition~\ref{lemma:ML}, next let us show that for any $(n,(M_1,\ldots,M_L))$ code that satisfies both asymptotically perfect reconstruction constraint \eqref{eq:M-Rec} and the perfect secrecy constraint \eqref{eq:M-Per} and any $\alpha=1,\ldots,L-N$, we have
\begin{equation}
\sum_{l=1}^{L}H(X_l) \geq \sum_{k=1}^{\alpha}\frac{nL}{k}H(S_k)+\Delta_\alpha-\sum_{k=1}^{\alpha}nL\delta_k(n,\epsilon)
\label{eq:clm}
\end{equation}
where
\begin{equation}
\Delta_\alpha := \frac{L}{
\left(
\begin{array}{c}
  L   \\
  N   
\end{array}
\right)
\left(
\begin{array}{c}
  L-N   \\
  \alpha   
\end{array}
\right)
}
\sum_{A \in \Omega_L^{(N)}}\sum_{D \in \Omega_{L \setminus A}^{(\alpha)}}\frac{H(X_D|S_1^n,\ldots,S_{\alpha}^n,X_A)}{\alpha}.
\end{equation}

We shall consider proof by induction, where the induction is on $\alpha$. First consider the base case with $\alpha=1$. Let $A \in \Omega_L^{(N)}$ and let $l \in \Omega_L \setminus A$. Applying Proposition~\ref{lemma:ML} with $k=1$, we have
\begin{equation}
H(X_l) \geq nH(S_1)+H(X_l|S_1^n,X_A)-n\delta_1(n,\epsilon). \label{T802}
\end{equation}
Averaging \eqref{T802} over all $l \in \Omega_L\setminus A$ and all $A \in \Omega_L^{(N)}$, we have
\begin{equation}
\frac{1}{
\left(
\begin{array}{c}
  L   \\
  N   
\end{array}
\right)
\left(
\begin{array}{c}
  L-N   \\
  1   
\end{array}
\right)}\sum_{A\in \Omega_L^{(N)}}\sum_{l \in \Omega_L\setminus A}H(X_l) \geq nH(S_1)+\frac{1}{L}\Delta_1-n\delta_1(n,\epsilon).
\end{equation}
Note that 
\begin{equation}
\frac{1}{\left(
\begin{array}{c}
  L   \\
  N   
\end{array}
\right)
\left(
\begin{array}{c}
  L-N   \\
  1   
\end{array}
\right)}\sum_{A\in \Omega_L^{(N)}}\sum_{l \in \Omega_L\setminus A}H(X_l) = \frac{
1}{L}\sum_{l=1}^{L}H(X_l).
\end{equation}
We thus have
\begin{equation}
\sum_{l=1}^LH(X_l) \geq nL H(S_1)+\Delta_1-nL\delta_1(n,\epsilon)
\end{equation}
which completes the proof of the base case.

Now assume that \eqref{eq:clm} holds for $\alpha-1$ for some $2 \leq \alpha \leq L-N$. Based on this assumption, next we show that \eqref{eq:clm} also holds for $\alpha$. By the classical subset inequality of Han \cite[Ch.~17.6]{CT-B06}, for any $A \in \Omega_L^{(N)}$ we have
\begin{equation}
\begin{array}{l}
\frac{1}{\left(
\begin{array}{c}
  L-N   \\
  \alpha-1   
\end{array}
\right)}\sum_{D \in \Omega_{L \setminus A}^{(\alpha-1)}}\frac{H(X_D|S_1^n,
\ldots,S_{\alpha-1}^n,X_A)}{\alpha-1}\\
\hspace{100pt} \geq \; 
\frac{1}{\left(
\begin{array}{c}
  L-N   \\
  \alpha   
\end{array}
\right)}\sum_{D \in \Omega_{L \setminus A}^{(\alpha)}}\frac{H(X_D|S_1^n,
\ldots,S_{\alpha-1}^n,X_A)}{\alpha}.
\end{array}
\end{equation}
It follows that
\begin{equation}
\Delta_{\alpha-1} \geq \frac{L}{
\left(
\begin{array}{c}
  L   \\
  N   
\end{array}
\right)
\left(
\begin{array}{c}
  L-N   \\
  \alpha   
\end{array}
\right)
}
\sum_{A \in \Omega_L^{(N)}}\sum_{D \in \Omega_{L \setminus A}^{(\alpha)}}\frac{H(X_D|S_1^n,
\ldots,S_{\alpha-1}^n,X_A)}{\alpha}.
\label{eq:Han}
\end{equation}
By Proposition~\ref{lemma:ML}, for any $A \in \Omega_L^{(N)}$ and any $D \in \Omega_{L \setminus A}^{(\alpha)}$ we have
\begin{equation}
H(X_D|S_1^n,\ldots,S_{\alpha-1}^n,X_A) \geq nH(S_\alpha)+H(X_D|S_1^n,\ldots,S_\alpha^n,X_A)-n\delta_n(\alpha,\epsilon).
\label{T2004}
\end{equation}
Substituting \eqref{T2004} into \eqref{eq:Han} gives
\begin{eqnarray}
\Delta_{\alpha-1} & \geq & \frac{L}{
\left(
\begin{array}{c}
  L   \\
  N   
\end{array}
\right)
\left(
\begin{array}{c}
  L-N   \\
  \alpha   
\end{array}
\right)
}
\sum_{A \in \Omega_L^{(N)}}\sum_{D \in \Omega_{L \setminus A}^{(\alpha)}}\nonumber\\
&& \hspace{120pt} \frac{nH(S_{\alpha})+H(X_D|S_1^n,
\ldots,S_{\alpha}^n,X_A)-n\delta_\alpha(n,\epsilon)}{\alpha}\\
& = & \frac{nL}{\alpha}H(S_{\alpha})+\Delta_{\alpha}-nL\delta_\alpha(n,\epsilon).
\label{eq:Delta}
\end{eqnarray}
By the induction assumption,
\begin{eqnarray}
\sum_{l=1}^{L}H(X_l) & \geq & \sum_{k=1}^{\alpha-1}\frac{nL}{k}H(S_k)+\Delta_{\alpha-1}-\sum_{k=1}^{\alpha-1}nL\delta_k(n,\epsilon)\\
& \geq & \sum_{k=1}^{\alpha-1}\frac{nL}{k}H(S_k)+\left(\frac{nL}{\alpha}H(S_{\alpha})+\Delta_{\alpha}-nL\delta_\alpha(n,\epsilon)\right)-\sum_{k=1}^{\alpha-1}nL\delta_k(n,\epsilon)\\
& = & \sum_{k=1}^{\alpha}\frac{nL}{k}H(S_k)+\Delta_{\alpha}-\sum_{k=1}^{\alpha}nL\delta_k(n,\epsilon).
\end{eqnarray}
This completes the proof of the induction step and hence \eqref{eq:clm}.

Finally, let $\alpha=L-N$ in \eqref{eq:clm}. For any admissible rate tuple $(R_1,\ldots,R_L)$ and any $\epsilon>0$, we have
\begin{eqnarray}
n\sum_{l=1}^{L}(R_l+\epsilon) & \geq & \sum_{l=1}^{L}H(X_l)\\
& \geq & \sum_{k=1}^{L-N}\frac{nL}{k}H(S_k)+\Delta_{L-N}-\sum_{k=1}^{L-N}nL\delta_k(n,\epsilon)\\
& \geq & \sum_{k=1}^{L-N}\frac{nL}{k}H(S_k)-\sum_{k=1}^{L-N}nL\delta_k(n,\epsilon)\label{T700}
\end{eqnarray}
where \eqref{T700} follows from the fact that $\Delta_{L-N} \geq 0$. Divide both sides of \eqref{T700} by $n$ and let $n \rightarrow \infty$ and $\epsilon \rightarrow 0$. Note that $\delta_k(n,\epsilon) \rightarrow 0$ in the limit as $n \rightarrow \infty$ and $\epsilon \rightarrow 0$ for all $k=1,\ldots,L-N$. We thus have
\begin{equation}
\sum_{l=1}^LR_l \geq \sum_{k=1}^{L-N}\frac{L}{k}H(S_k)
\end{equation}
for any admissible rate tuple $(R_1,\ldots,R_L)$. This completes the proof of Theorem~\ref{thm:S-SMDC1}.

\section{Concluding Remarks} \label{sec:Con}
This paper considered the problem of S-SMDC, which is a natural (perhaps also the simplest) extension of the classical SMDC problem \cite{Roc-Thesis92,Yeu-IT95,RYH-IT97,YZ-IT99} to the secrecy communication setting. First, the problem of encoding individual sources, i.e., the S-SSDC problem, was studied. A precise characterization of the entire admissible rate region was established via a connection to the problem of secure coding over a three-layer WN \cite{CY-ITS} and utilizing some basic polyhedral structure of the admissible rate region. Building on this result, it was then shown that the simple coding strategy of separately encoding individual sources at the encoders (superposition coding) can achieve the minimum sum rate for the general S-SMDC problem.

Based on the result of Theorem~\ref{thm:S-SMDC1} (and the fact that superposition coding can achieve the entire admissible rate region for the classical SMDC problems without secrecy constraints), it is very tempting to conjecture that superposition coding can in fact achieve the entire admissible rate region for the general S-SMDC problem. In Appendix~\ref{app:B}, we verify that this is indeed the case for the simplest nontrivial S-SMDC problem: the $(3,1)$ S-SMDC problem. Our proof relies on an \emph{explicit} characterization of the superposition coding rate region via a Fourier-Motzkin elimination procedure. The optimality of superposition coding is then established by carefully using the results of Proposition~\ref{lemma:ML}.

Extending such a proof strategy to the general $(L,N)$ S-SMDC problem, however, faces a number of challenges. To begin with, the complexity of Fourier-Motzkin elimination procedure grows unboundedly as the total number of encoders $L$ increases. Thus, establishing an explicit characterization of the superposition coding rate region for the general $(L,N)$ S-SMDC problem appears to be very difficult. An alternative strategy is to look for an \emph{implicit} characterization of the superposition coding rate region using \emph{optimal $\alpha$-resolutions}, similar to that \cite{YZ-IT99} for the classical SMDC problem without any secrecy constraints. In fact, note from Theorem~\ref{thm:S-SSDC} that the admissible rate region of an $(L,N,m)$ S-SSDC problem depends on the parameters $N$ and $m$ only via its difference $m-N$. As mentioned previously in Sec.~\ref{sec:SL}, when $N=0$, the $(L,N,m)$ S-SSDC problem reduces to the classical $(L,m)$ SSDC problem without any secrecy constraints. Thus, the admissible rate region of the $(L,N,N+k)$ S-SSDC problem with source $S_k$ is \emph{identical} to that of the classical $(L,k)$ SSDC problem with the same source. As a result,  the superposition coding rate region of the $(L,N)$ S-SMDC problem with sources $(S_1,\ldots,S_{L-N})$ is \emph{identical} to the superposition coding rate region of the classical SMDC problem with a total of $L$ encoders and sources $(S_1,\ldots,S_{L})$ where the entropy rate of the source $H(S_l)=0$ for $l=L-N+1,\ldots,L$. Based on this observation, the $\alpha$-resolution characterization of the superposition coding rate region for the general SMDC problem can be directly translated to the S-SMDC problem. It remains to see whether the properties provided in \cite{YZ-IT99} on optimal $\alpha$-resolutions are sufficient for establishing the optimality of superposition coding for the general S-SMDC problem. This problem is currently under our investigations.

\appendix
\section{Proof of Lemma~\ref{lemma:S-SSDC}}\label{app:A}
Let $D \in \Omega_L^{(m-N)}$ and let $A \in \Omega_{L\setminus D}^{(N)}$. Since $A\cap D =\emptyset$, we have $|D \cup A|=N+(m-N)=m$. For any $(n,(M_1,\ldots,M_L))$ code that satisfies both asymptotically perfect reconstruction constraint \eqref{eq:S-Rec} and the perfect secrecy constraint \eqref{eq:S-Per}, we have by Fano's inequality
\begin{equation}
H(S^n|X_D,X_A)\leq n\delta(n,\epsilon) \label{W1}
\end{equation}
where
\begin{equation}
\delta(n,\epsilon)=1/n+\epsilon\log|\mathcal{S}|
\end{equation}
and
\begin{equation}
H(S^n|X_A)=H(S^n).\label{W2}
\end{equation}
For any admissible rate tuple $(R_1,\ldots,R_L)$ and any $\epsilon>0$, we have
\begin{eqnarray}
n\sum_{l \in D} (R_l+\epsilon) & \geq & \sum_{l \in D} H(X_l)\label{A100}\\
& \geq & H\left(X_D\right)\label{A101}\\
& \geq & H\left(X_D|X_A\right)\label{A102}\\
& \geq & H\left(X_D|X_A\right)+H\left(S^n|X_D,X_A\right)-n\delta(n,\epsilon)\label{A103}\\
& = & H\left(X_D,S^n|X_A\right)-n\delta(n,\epsilon)\\
& = & H\left(S^n|X_A\right)+H\left(X_D|S^n,X_A\right)-n\delta(n,\epsilon)\\
& \geq & H\left(S^n|X_A\right)-n\delta(n,\epsilon)\\
& = & H\left(S^n\right)-n\delta(n,\epsilon)\label{A104}\\
& = & nH(S)-n\delta(n,\epsilon)\label{A105}
\end{eqnarray}
where \eqref{A101} follows from the independence bound on entropy, \eqref{A102} follows from the fact that conditioning reduces entropy, \eqref{A103} follows from \eqref{W1}, \eqref{A104} follows from \eqref{W2}, and \eqref{A105} follows from the fact that the source $S$ is memoryless. Divide both sides of \eqref{A105} by $n$ and let $n \rightarrow \infty$ and $\epsilon \rightarrow 0$. Note that $\delta(n,\epsilon) \rightarrow 0$ in the limit as $n \rightarrow \infty$ and $\epsilon \rightarrow 0$. We have from \eqref{A105} that 
\begin{equation}
\sum_{l \in D} R_l \geq H(S), \quad \forall D \in \Omega_L^{(m-N)}.
\end{equation} 
This completes the proof of Lemma~\ref{lemma:S-SSDC}.

\section{The Admissible Rate Region of the $(3,1)$ S-SMDC Problem}\label{app:B}
In this appendix, we show that superposition coding can achieve the entire admissible rate region for the $(3,1)$ S-SMDC problem (the simplest nontrivial S-SMDC problem). The result is summarized in the following theorem.

\begin{theorem}\label{thm:MS2} 
Superposition coding can achieve the entire admissible rate region for the $(3,1)$ S-SMDC problem, which  is given by the collection of all rate triples $(R_1,R_2,R_3)$ satisfying
\begin{equation}
\begin{array}{rcl}
R_1 & \geq & H(S_1)\\
R_2 & \geq & H(S_1)\\
R_3 & \geq & H(S_1)\\
R_1+R_2 & \geq & 2H(S_1)+H(S_2)\\
R_2+R_3 & \geq & 2H(S_1)+H(S_2)\\
R_3+R_1 & \geq & 2H(S_1)+H(S_2).
\label{eq:MS2}
\end{array}
\end{equation}
\end{theorem}

\begin{proof}
\emph{Achievability.} Consider the superposition coding scheme that separately encodes the sources $S_1$ and $S_2$ using the $(3,1,2)$ and $(3,1,3)$ S-SSDC codes, respectively. By Theorem~\ref{thm:S-SSDC}, the admissible rate region for the $(3,1,2)$ S-SSDC problem is given by all rate triples $(R_1^{(1)},R_2^{(1)},R_3^{(1)})$ satisfying
\begin{equation}
\begin{array}{rcl}
R_1^{(1)} & \geq & H(S_1)\\
R_2^{(1)} & \geq & H(S_1)\\
R_3^{(1)} & \geq & H(S_1)
\end{array}
\label{FM-S}
\end{equation}
and the admissible rate region for the $(3,1,3)$ S-SSDC problem is given by all rate triples $(R_1^{(2)},R_2^{(2)},R_3^{(2)})$ satisfying
\begin{equation}
\begin{array}{rcl}
R_1^{(2)} & \geq & 0\\
R_2^{(2)} & \geq & 0\\
R_3^{(2)} & \geq & 0\\
R_1^{(2)}+R_2^{(2)} & \geq & H(S_2)\\
R_2^{(2)}+R_3^{(2)} & \geq & H(S_2)\\
R_3^{(2)}+R_1^{(2)} & \geq & H(S_2).
\end{array}
\end{equation}
Following \eqref{eq:SC}, all rate triples $(R_1,R_2,R_3)$ as given by
\begin{equation}
R_l = R_l^{(1)}+R_l^{(2)}, \quad \forall l=1,2,3 \label{FM-E}
\end{equation}
are admissible via superposition coding. Using Fourier-Motzkin elimination to eliminate $R_l^{(k)}$, $l=1,2,3$ and $k=1,2$, from \eqref{FM-S}--\eqref{FM-E}, we obtain the explicit characterization of the superposition coding rate region for the $(3,1)$ S-SMDC problem as expressed by \eqref{eq:MS2}.

\emph{The converse.} Next, we establish the optimality of superposition coding by proving that every inequality in \eqref{eq:MS2} must hold for \emph{all} admissible rate triples $(R_1,R_2,R_3)$ for the $(3,1)$ S-SMDC problem. Let
\begin{equation} 
a \oplus b := \left\{
\begin{array}{ll}
a+b,& \mbox{if}\; a+b \leq 3\\
a+b-3,& \mbox{otherwise}.
\end{array}
\right.
\end{equation}
For any admissible rate triple $(R_1,R_2,R_3)$, any $l=1,2,3$, and any $\epsilon>0$, we have
\begin{eqnarray}
n(R_l+\epsilon) & \geq & H(X_l) \label{T99}\\
& \geq & H(X_l|X_{l \oplus 1}) \label{T100}\\
& \geq & nH(S_1)+H(X_l|S_1^n,X_{l\oplus 1})-n\delta_1(n,\epsilon)\label{T101}\\
& \geq & nH(S_1)-n\delta_1(n,\epsilon)\label{T102}
\end{eqnarray} 
and
\begin{eqnarray}
&&n(R_l+R_{l \oplus 1}+2\epsilon) \nonumber\\
& & \hspace{20pt} \geq \; H(X_l)+H(X_{l\oplus 1})\label{T199}\\
& & \hspace{20pt} \geq \;  H(X_l|X_{l\oplus 1})+H(X_{l\oplus 1}|X_{l\oplus 2})\label{T200}\\
& & \hspace{20pt} \geq \;  2nH(S_1)+H(X_l|S_1^n,X_{l\oplus 1})+H(X_{l\oplus 1}|S_1^n,X_{l\oplus 2})-2n\delta_1(n,\epsilon)
\label{T201}\\
& & \hspace{20pt} \geq \;  2nH(S_1)+H(X_l|S_1^n,X_{l\oplus 1},X_{l\oplus 2})+H(X_{l\oplus 1}|S_1^n,X_{l\oplus 2})-2n\delta_1(n,\epsilon)
\label{T202}\\
& & \hspace{20pt} = \;  2nH(S_1)+H(X_l,X_{l\oplus 1}|S_1^n,X_{l\oplus 2})-2n\delta_1(n,\epsilon)\\
& & \hspace{20pt} \geq \;  2nH(S_1)+\left(nH(S_2)+H(X_l,X_{l\oplus 1}|S_1^n,S_2^n,X_{l\oplus 2})-n\delta_2(n,\epsilon)\right)-2n\delta_1(n,\epsilon)\label{T203}\\
& & \hspace{20pt} \geq \;  2nH(S_1)+nH(S_2)-n\delta_2(n,\epsilon)-2n\delta_1(n,\epsilon).\label{T204}
\end{eqnarray} 
Here, \eqref{T100}, \eqref{T200} and \eqref{T202} follow from the fact that conditioning reduces entropy, and \eqref{T101}, \eqref{T201} and \eqref{T203} follow from Proposition~\ref{lemma:ML}. Dividing both sides of \eqref{T102} and \eqref{T204} by $n$ and letting $n \rightarrow \infty$ and $\epsilon \rightarrow 0$ complete the proof of the converse part of the theorem.
\end{proof}

\end{document}